\documentclass[]{article}
\usepackage[a4paper]{geometry}
\usepackage[utf8]{inputenc}
\usepackage[T1]{fontenc} 
\usepackage[english]{babel}
\usepackage{amsmath,amssymb,latexsym,eufrak,euscript}
\usepackage{subfigure,pstricks,pst-node,pst-coil}

\usepackage{url,tikz}
\usepackage{pgflibrarysnakes}

\usepackage{multicol}

\usetikzlibrary{arrows}
\usetikzlibrary{automata}
\usetikzlibrary{snakes}
\usetikzlibrary{shapes}

\newcommand{\stirlingtwo}[2]{\genfrac{\lbrace}{\rbrace}{0pt}{}{#1}{#2}}

\newtheorem{definition}{Definition}
\newtheorem{theorem}[definition]{Theorem}
\newtheorem{proposition}[definition]{Proposition}
\newtheorem{example}[definition]{Example}

\newenvironment{proof}
{ \noindent {\sc Proof.\/}  }
{\null  \hfill $\Box$ \par\medskip \vspace{1em}}

\def \A {\mathcal{A}}
\def \T {\mathfrak{T}}
\def \M {\mathcal{M}}

\begin{document}
\label{firstpage}

  \title{Random Generation and Enumeration of Accessible Deterministic
    Real-time Pushdown Automata} 
\author{Pierre-Cyrille H\'eam \and Jean-Luc Joly}
 
\maketitle

\begin{center}{FEMTO-ST, CNRS UMR 6174, Universit\'e de Franche-Comt\'e, INRIA\\ 
    16 route de Gray,
    25030 Besan\c con Cedex, France \\
} 
\end{center}

\begin{abstract}
This paper presents a general framework for the uniform random generation
of deterministic real-time accessible pushdown automata. A polynomial time
algorithm to randomly generate a pushdown automaton having a fixed stack
operations total size is proposed. The influence of the accepting condition
(empty stack, final state) on the reachability of the generated automata is
investigated.
\end{abstract}

\section{Introduction}


Finite automata, of any kind, are widely used for their algorithmic
properties in many fields of computer science like model-checking,
pattern matching and machine learning. Developing new efficient algorithms
for finite automata is therefore a challenging problem still addressed by
many recent papers. New algorithms are frequently motivated by improvement
of worst cases bound. However, several examples, such as sorting algorithms,
primality testing or solving linear problems, show that worst case
complexity is not always the right way to evaluate the practical performance
of an algorithm. When benchmarks are not available, random testing, with a
controlled distribution, represents an efficient mean of performance
testing. In this context, the problem of uniformly generating finite
automata is a challenging problem.

This paper tackles with the problem of the uniform random generation of
real-time deterministic pushdown automata. Using classical combinatorial
techniques, we expose how to extend existing works on the generation of
finite deterministic automata to pushdown automata. More precisely, we show
in Section~\ref{sec:random} how to uniformly generate and enumerate (in the
complete case) accessible real-time deterministic pushdown automata. In
Section~\ref{sec:PDA}, it is shown that using a rejection algorithm it is
possible to efficiently generate pushdown automata that don't accept an
empty language. The influence of the accepting condition (final state or
empty-stack) on the reachability
of the generated pushdown automata is also experimentally studied  in
Section~\ref{sec:PDA}.

\paragraph{Related work.}  The 
enumeration of deterministic finite automata has been first investigated
in~\cite{Vyssotsky} and was applied to several subclasses of deterministic
finite
automata~\cite{Korshunov,DBLP:journals/eik/Korshunov86,robinson,DBLP:journals/dam/Liskovets06}. 
The uniform random generation of accessible deterministic complete automata
was initially proposed in~\cite{thesecril} for two-letter alphabets and the
approach was extended to larger alphabets
in~\cite{DBLP:journals/tcs/ChamparnaudP05}. Better algorithms can be found
in~\cite{DBLP:journals/tcs/BassinoN07,DBLP:conf/stacs/CarayolN12}. The
random generation of possibly incomplete automata is analyzed
in~\cite{incomplet}. The recent paper~\cite{DBLP:conf/wia/CarninoF11}
presents how to use Monte-Carlo approaches to generate deterministic acyclic
automata. As far as we know, the only work focusing on the random generation
of deterministic transducers is~\cite{DBLP:journals/tcs/HeamNS10}. This work
can be applied to the random generation of deterministic real-time
pushdown automata that can be possibly incomplete. However the requirement to fix the size of the stack operation on each transition, represents a major restriction.
The reader interested in the random generation of deterministic automata is
 referred to the survey~\cite{DBLP:conf/mfcs/Nicaud14}.

\section{Formal Background}\label{sec:bg}

We assume that the reader is familiar with classical notions on formal
languages. For more information on automata theory or on pushdown automata
the reader is referred to~\cite{Hopcroft} or to~\cite{Saka}. For a general
reference on random generation and enumeration of combinatorial structures
see~\cite{DBLP:journals/tcs/FlajoletZC94}. For any word $w$ on an alphabet $\Sigma$, $|w|$ denotes
its length. The empty word is denoted~$\varepsilon$.
The cardinal of a finite set $X$ is denoted $|X|$.

\paragraph{Deterministic Finite Automata.}
A {\it deterministic finite automaton} on $\Sigma$ is a tuple
$(Q,\Sigma,\delta,q_{\rm init},F)$ where $Q$ is a finite set of states,
$\Sigma$ is a finite alphabet, $q_{\rm init}\in Q$ is the initial state,
$F\subseteq Q$ is the set of final states and $\delta$ is a partial function
from $Q\times\Sigma$ into $Q$. If $\delta$ is not partial, i.e., defined for
each $(q,a)\in Q\times\Sigma$, the automaton is said {\it complete}. A
triplet of the form $(q,a,\delta(p,a))$ is called a {\it transition}. A
finite automaton is graphically represented by a labeled finite graph whose
vertices are the states of the automaton and edges are the transitions. A
deterministic finite automaton is {\it accessible} if for each state $q$
there exists a path from the initial state to $q$. Two finite automata
$(Q_1,\Sigma,\delta_1,q_{\rm init1},F_1)$ and $(Q_2,\Sigma,\delta_2,q_{\rm
init2},F_2)$ are isomorphic if they are identical up to the state's names,
formally if there exists a one-to-one function $\varphi$ from $Q_1$ into
$Q_2$ such that (1) $\varphi(q_{\rm init1})=q_{\rm init2}$,  (2)
$\varphi(F_1)=F_2$, and (3) $\delta_1(q,a)=p$ iff
$\delta_2(\varphi(q),a)=\varphi(p)$.

\paragraph{Pushdown Automata.}
A {\it real-time deterministic pushdown automaton}, RDPDA for short, is a
tuple $(Q,\Sigma,\Gamma,Z_{\rm init},\delta,q_{\rm init},F)$ where $Q$ is a
finite set of states, $\Sigma$ and $\Gamma$ are finite disjoint alphabets,
$q_{\rm init}\in Q$ is the initial state, $F\subseteq Q$ is the set of final
states, $Z_{\rm init}$ is the initial stack symbol and $\delta$ is a partial
function from $Q\times(\Sigma\times\Gamma)$ into $Q\times\Gamma^*$. If
$\delta$ is not partial, i.e., defined for each $(q,(a,X))\in Q\times(\Sigma\times\Gamma)$, the
RDPDA is said to be {\it complete}. 
A triplet of the form $(q,(a,X),w,p)$ with
$\delta(q,(a,X))=(p,w)$ is called a {\it transition} and $w$ is the {\it
output} of the transition. The {\it output size} of a transition
$(q,(a,X),w,p)$ is the length of $w$.  The {\it output size} of
a RDPDA is the sum of the sizes of its transitions.
The {\it underlying automaton} of an
RDPDA, is the finite automaton $(Q,\Sigma\times\Gamma,\delta^\prime,q_{\rm
init},F)$, with $\delta^\prime(q,(a,X))=p$ iff $\delta((q,(a,X)))=(p,w)$
for some $w\in\Gamma^*$. An RDPDA is {\it accessible} if its underlying
automaton is accessible. A transition whose output is $\varepsilon$ is
called a {\it pop transition}. An example of a complete accessible RDPDA is
depicted in Fig.~\ref{figAutoPile1}. The related underlying  finite
automaton is depicted in Fig.~\ref{figUnder}.

\begin{figure}[t!]
\centering
\subfigure{
\begin{tikzpicture}
\node (0) [state,fill=black!20,initial,initial text=] at (0,0) {$0$};
\node (1) [state,fill=black!20, accepting]at (3,0) {$1$};

\path[->,>=triangle 90] (0) [bend left] edge[above] node {$(a,X),Z$} (1);
\path[->,>=triangle 90] (1) [bend left] edge[below, pos=0.3] node
     {\begin{tabular}{c}$(b,X),\varepsilon$\\$(a,Z),XZX$\end{tabular}} (0);

\path[->,>=triangle 90] (0) [loop above] edge[above] node
     {\begin{tabular}{c}$(a,Z),ZZX$\\$(b,Z),ZX$\end{tabular}} ();
\path[->,>=triangle 90] (0) [loop below] edge[below] node
     {$(b,X),X$} (0);

\path[->,>=triangle 90] (1) [loop above] edge[above] node
     {\begin{tabular}{c}$(a,X),XX$\\$(b,Z),\varepsilon$\end{tabular}} ();

\node (txt) at (7.5,0.3){
$\begin{array}{rcl}
Q &= & \{0,1\}\\
\Sigma&=&\{a,b\}\\
\Gamma&=&\{Z,X\}\\
\delta&=&\{(0,(a,X))\mapsto (1,Z),(0,(b,X))\mapsto (0,X)\\
&&(0,(a,Z)) \mapsto (0,ZZX),(0,(b,Z))\mapsto (0,ZX)\\
&&(1,(a,X))\mapsto (1,XX), (1,(b,X)) \mapsto (0,\varepsilon)\\
&&(1,(a,Z))\mapsto(0, XZX),(1,(b,Z))\mapsto (1,\varepsilon)\\
q_{\rm init}&=&0,\quad Z_{\rm init}=Z,\quad F=\{1\}
\end{array}$};
\end{tikzpicture}
}
\caption{$P_{\rm toy}$, a complete RDPDA.\label{figAutoPile1}}
\end{figure}
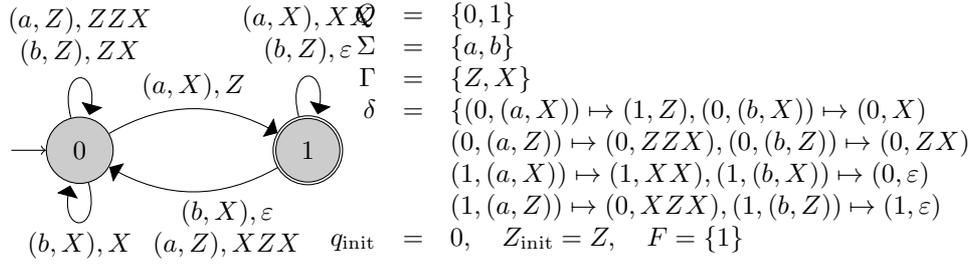

A {\it configuration} of a {RDPDA} is an element of $Q\times\Gamma^*$. The
{\it initial configuration} is $(q_{\rm init},Z_{\rm init})$. Two
configurations $(q_1,w_1)$ and $(q_2,w_2)$ are $a$-consecutive, denoted
$(q_1,w_1)\models_a(q_2,w_2)$ if the following conditions are satisfied:
\begin{itemize}
\item $w_1\neq \varepsilon$, and let $w_1=w_3X$ with $X\in \Gamma$,
\item  $\delta(q_1,(X,a))=(q_2,w_4)$ and $w_2=w_3w_4$.
\end{itemize}
Two configurations are {\it consecutive} if there is a letter $a$ such that there
are $a$-consecutive. A state $p$ of a RDPDA is {\it reachable} if there
exists a sequence of consecutive configurations $(p_1,w_1),\ldots,(p_n,w_n)$
such that $(p_1,w_1)$ is the initial configuration and $p_n=p$. Moreover, if
$w_n=\varepsilon$, $p$ is said to be {\it reachable with an empty stack}. A
RDPDA is {\it reachable} if all its states are reachable. Consider for
instance the RDPDA of Fig.~\ref{fig:acceptation}, where the initial stack
symbol is $X$. State $3$ is not reachable since the transition from $0$ to
$3$ cannot be fired. State $1$ is reachable with an empty stack. State $2$ is
reachable, but not reachable with an empty stack. Note that a reachable
state is accessible, but the converse is not true in general: accessibility
is a notion defined on the underlying finite automaton. 

 The configurations $(1,XZ)$ and $(2,XXZX)$ are $a$-consecutive on the RDPDA
depicted in Fig.~\ref{figAutoPile1}. 

\begin{figure}[h!]
\centering
\subfigure{
\begin{tikzpicture}
\node (0) [state,fill=black!20,initial,initial text=] at (0,0) {$0$};
\node (1) [state,fill=black!20, accepting]at (3,0) {$1$};

\path[->,>=triangle 90] (0) [bend left] edge[above] node {$(a,X)$} (1);
\path[->,>=triangle 90] (1) [bend left] edge[below, pos=0.3] node
     {\begin{tabular}{c}$(b,X)$\\$(a,Z)$\end{tabular}} (0);

\path[->,>=triangle 90] (0) [loop above] edge[above] node
     {\begin{tabular}{c}$(a,Z)$\\$(b,Z)$\end{tabular}} ();
\path[->,>=triangle 90] (0) [loop below] edge[below] node
     {$(b,X)$} (0);

\path[->,>=triangle 90] (1) [loop above] edge[above] node
     {\begin{tabular}{c}$(a,X)$\\$(b,Z)$\end{tabular}} ();

\node (txt) at (7.5,0.3){
$\begin{array}{rcl}
Q &= & \{0,1\}\\
\text{The alphabet}& \text{is} &\{a,b\}\times\{Z,X\}\\
\delta^\prime&=&\{(0,(a,X))\mapsto 1,(0,(b,X))\mapsto 0\\
&&(0,(a,Z)) \mapsto 0,(0,(b,Z))\mapsto 0\\
&&(1,(a,X))\mapsto 1, (1,(b,X)) \mapsto 0\\
&&(1,(a,Z))\mapsto 0,(1,(b,Z))\mapsto 1\\
&&q_{\rm init}=0\quad F=\{1\}
\end{array}$};
\end{tikzpicture}
}
\caption{Underlying automaton of $P_{\rm toy}$.\label{figUnder}}
\end{figure}
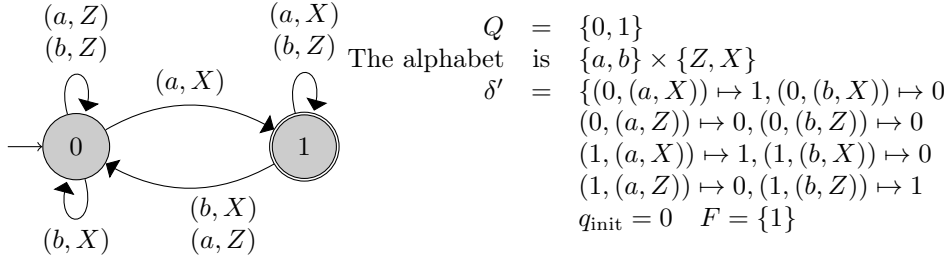
There are three main kinds of accepting
conditions for a word $u=a_1\ldots a_k\in\Sigma^*$ by an RDPDA:
\begin{itemize}
\item Under the {\it empty-stack
 condition},  $u$ is accepted if there exists 
configurations $c_1,\ldots,c_{k+1}$ such that $c_1$ is the initial
configuration, $c_i$ and $c_{i+1}$ are $a_i$-consecutive, and $c_{k+1}$ is
of the form $(q,\varepsilon)$. 
\item Under the {\it final-state
 condition}, $u$ is accepted  if there exists 
configurations $c_1,\ldots,c_{k+1}$ such that $c_1$ is the initial
configuration, $c_i$ and $c_{i+1}$ are $a_i$-consecutive, and $c_{k+1}$ is
of the form $(q,w)$, with $q\in F$.
 \item Under the {\it final-state and empty-stack
 condition}, $u$ is accepted  if there exists 
configurations $c_1,\ldots,c_{k+1}$ such that $c_1$ is the initial
configuration, $c_i$ and $c_{i+1}$ are $a_i$-consecutive, and $c_{k+1}$ is
of the form $(q,\varepsilon)$, with $q\in F$.
 \end{itemize}

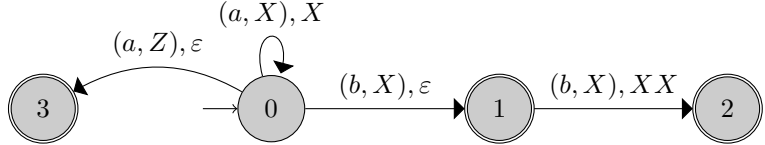
\begin{figure}[h!]
\begin{center}
\begin{tikzpicture}
\node (0) [state,fill=black!20,initial,initial text=] at (0,0) {$0$};
\node (1) [state,fill=black!20, accepting]at (3,0) {$1$};
\node (2) [state,fill=black!20, accepting]at (6,0) {$2$};
\node (3) [state,fill=black!20, accepting]at (-3,0) {$3$};

\path[->,>=triangle 90] (0) [] edge[above] node {$(b,X),\varepsilon$} (1);
\path[->,>=triangle 90] (1) [] edge[above] node
     {$(b,X),XX$} (2);

\path[->,>=triangle 90] (0) [loop above] edge[above] node
     {$(a,X),X$} ();
\path[->,>=triangle 90] (0) [bend right] edge[above] node
     {$(a,Z),\varepsilon$} (3);

\end{tikzpicture}
\end{center}
\caption{Acceptance conditions.}\label{fig:acceptation}
\end{figure}
Consider for instance the RDPDA of Fig.~\ref{fig:acceptation}, where the
initial stack symbol is $X$. With the empty-stack condition only the word
$b$ is accepted, as well as for the empty-stack and final state condition.
With the final state condition, the accepted language is $a^*(b+bb)$.

Two RDPDA $(Q_1,\Sigma,\Gamma,\delta_1,q_{\rm init1},F_1)$
and $(Q_2,\Sigma,\Gamma,\delta_2,q_{\rm init2},F_2)$ are isomorphic if
there exists a one-to-one function $\varphi$ from $Q_1$ into $Q_2$ such that
(i) $\varphi(q_{\rm init1})=q_{\rm init2}$, and (ii) $\varphi(F_1)=F_2$, and
 $$(iii)\quad \delta_1(q,(a,X))=(p,w)\quad \text{iff}\quad \delta_2(\varphi(q),(a,X))=(\varphi(p),w).$$
Note that if two RDPDA are isomorphic, then their underlying automata are
isomorphic too.

\paragraph{Generating Functions.}
A {\it combinatorial class} is a class $\mathfrak{C}$ of objects associated
with a size function $|.|$ from $\mathfrak{C}$ into $\mathbb{N}$ such that
for any integer $n$ there are finitely many elements of $\mathfrak{C}$ of
size $n$. The ordinary generating function for $\mathfrak{C}$ is
$C(z)=\sum_{c\in \mathfrak{C}} z^{|c|}$. The $n$-th coefficient of $C(z)$ is
exactly the number of objects of size $n$ and is denoted $[z^n]C(z)$. The
reader is referred to~\cite{FSbook} for the general methodology of analytic
combinatorics, ans especially the use of generating functions to count
objects. The following result~\cite[Theorem~VIII.8]{FSbook} will be useful
in this paper.

\begin{theorem}\label{thm:VIII.8}
Let $C(z)$ be an ordinary generating function satisfying: (1) $C(z)$ is
analytic at $0$ and have only positive coefficients, $(2)$ $C(0)\neq 0$ and
$(3)$ $C(z)$ is aperiodic. Let $R$ be the radius of convergence of $C(z)$
and $T=\lim_{x\to R^{-}} x\frac{C'(x)}{C(x)}$.
Let $\lambda\in ]0,T[$ and $\zeta$ be the  unique
solution of $x\frac{C'(x)}{C(x)}=\lambda$. Then, for $N=\lambda n$ an
integer, one has
$$[z^N]C(z)^n=\frac{C(\zeta)^n}{\zeta^{N+1}\sqrt{2\pi n\xi}}(1+o(1)),$$
where $\xi=\frac{d^2}{d z^2}\left(\log C(z)-\lambda\log (z)\right)|_{z=\zeta}.$
\end{theorem}
In the above theorem, $T$ is called the \textit{spread} of the $C(z)$. 

\paragraph{Rejection Algorithms.}
A rejection algorithm is a probabilistic algorithm to randomly
generate an element in a set $X$, using an algorithm $A$ for
generating an element of $Y$ in the simple way: repeat $A$ until it
returns an element of $X$. Such an algorithm is tractable if the
expected number of iterations can be kept under control (for instance is
fixed): the probability that an element of $Y$ is in $X$ has to be large
enough.

\paragraph{Random Generation.}
 The theory of Generating Functions provides an efficient way to randomly
 and uniformly generate an element of size $n$ of a combinatorial class
 $\mathfrak{C}$ using a recursive approach~\cite{FSbook}. It requires a
 $O(n^2)$ precomputation time and each random sample is obtained in time
 $O(n\log n)$. Another efficient way to uniformly generate element of
 $\mathfrak{C}$ is to use Boltzmann
 samplers~\cite{DBLP:journals/cpc/DuchonFLS04}: the random generation of an
 object with a size about $n$ (approximate sampling) is performed in $O(n)$,
 while the random generation of an object with a size exactly $n$ is
 performed in expected time $O(n^2)$ using a rejection algorithm (without
 precomputation). Boltzmann samplers are quite easy to implement but are
 restricted to a limited number of combinatorial constructions. They also
 require the evaluation of some generating functions at some values of the
 variable.

\section{Random Generation and Enumeration of RDPDA}\label{sec:random}
In this section, $\Sigma$ and $\Gamma$ are fixed disjoint alphabets of
 respective cardinals $\alpha$ and $\beta$. We denote by $\T_{s,n,m}$
 combinatorial class of the RDPDA (on $\Sigma,\Gamma$) with $n$ states,
 $s$ transitions and with an output size of $m$, up to isomorphism. Note
 that this class is well defined since two isomorphic RDPDA have the same
 output size. Let $\rho=\alpha\beta$.

\subsection{Enumeration of RDPDA}\label{sec:enum}

We are interested in the random generation of accessible RDPDA up to
isomorphism. The class of all isomorphic classes of accessible automata on
$\Sigma\times\Gamma$ with $n$ states and $s$ transitions is denoted
$\mathfrak{A}_{s,n}$. Let $\psi$ be the function from $\T_{s,n,m}$ into
$\mathfrak{A}_{s,n}$ mapping RDPDA to their underlying automata. The number
of elements of $\psi^{-1}(\A)$, where $\A$ is an element of
$\mathfrak{A}_{s,n}$, is the number of possible output labelling of the $s$
transitions of $\A$, which only depends on $s$ and $m$ and is independent of
$\A$. We denote by $c_{s,m}$ this number of labelings. The following
proposition is a direct consequence of the above remark.

\begin{proposition}\label{prop:1}
One has $|\T_{s,n,m}|=|\mathfrak{A}_{s,n}|\cdot c_{s,m}.$
\end{proposition}

Proposition~\ref{prop:1} is the base of the enumeration for RDPDA.
Let $M(z)$ denote the ordinary generating function of
elements of $\Gamma^*$. Then, using Proposition~\ref{prop:1} and classical
constructions on generating functions, one has
\begin{equation}\label{eq:enum0}
|\T_{s,n,m}(\M)|=|\mathfrak{A}_{s,n}|.[z^m]M(z)^s.
\end{equation}

Equation \eqref{eq:enum0} will be exploited for complete finite automata
using the following result of~\cite{Korshunov} -- see
also~\cite{DBLP:conf/wia/BassinoDN07}.

\begin{theorem}[\cite{Korshunov}] \label{thm:K}
There exists a constant $\gamma_{\rho}$,
such that 
$|\mathfrak{A}_{\rho n,n}|\sim n\gamma_\rho  \stirlingtwo{\rho n}{n},$
where $\stirlingtwo{x}{y}$ denotes the Stirling numbers of the second kind.
\end{theorem}
 
The case $s=\rho n$ corresponds to complete accessible automata.
We will particularly focus throughout this paper on  the complete case 
for a fixed average size $\lambda$ for the transitions: we assume that there is a
fixed $\lambda >0$ such that $m=\lambda s=\lambda n\rho$.
In this context, \eqref{eq:enum0} becomes

\begin{equation}\label{eq:enum-fin}
|\T_{s,n,m}(\M)|\sim n\gamma_\rho \stirlingtwo{\rho n}{n}[z^{\lambda
 n\rho}]M(z)^{n\rho}.
\end{equation}

Since $M(z)$ is the generating function of words on $\Gamma$, one has
$M(z)=\frac{1}{1-\beta z}$. 
Therefore (see~\cite{FSbook}), one has:
\begin{equation}\label{eq:fsn}
 [z^m]F_{s}(z)=\beta^m \dfrac{s
   (s+1)(s+2)\ldots(s+m-1)}{m!}.
\end{equation}

The generating function $F_{s}$ satisfies the hypotheses of
Theorem~\ref{thm:VIII.8}, with an infinite spread. Therefore for any strictly
positive $\lambda$ and any integer  $m=\lambda s$, one has

\begin{equation}\label{eq:app1}
\left[z^{\lambda s}\right]F_{s}(z)=\dfrac{C(\zeta)^{s}}{\zeta^{\lambda s+1}\sqrt{2\,\pi\,s\, \xi}}\left(1+o(1)\right),
\end{equation}
where $C(z)=\dfrac{1}{1-\beta z}$ and
$\zeta=\dfrac{\lambda}{\beta\left(\lambda+1\right)}$ is the unique solution
of $x\dfrac{C'(x)}{C(x)}=\lambda$. It follows
that
\begin{equation}\label{eq:equivc}
[z^{\lambda s}]F_{s}(z)  =  \dfrac{(\lambda+1)^{(\lambda+1)s+1}\beta^{\lambda s+1}}{\lambda^{\lambda s+1}\sqrt{2\pi  s \xi}}(1+o(1)).
\end{equation}

\noindent
In addition 
\begin{equation}
\xi=\dfrac{d^2}{d z^2}\left(\log C(z)-\lambda\log (z)\right)|_{z=\zeta}
=\dfrac { \left( \lambda+1 \right) ^{3}{\beta}^{2}}{\lambda}.
\end{equation}

\noindent
Consequently~\eqref{eq:equivc} can be rewritten as
\begin{equation}\label{eq:equivcc}
[z^{\lambda s}]F_{s}(z)  =\dfrac { \left( \lambda+1 \right) ^{(\lambda+1)\,s-1/2}{\beta}^{\lambda\,s}}
{{\lambda}^{\lambda\,s+1/2}\sqrt {2\,\pi \,s}}(1+o(1)).
\end{equation}
The following proposition is a direct consequence of the combination
of~\eqref{eq:equivcc}, Proposition~\ref{prop:1} and~ Theorem~\ref{thm:K}.

\begin{proposition}
The number $f_{\lambda,n}$ of complete accessible RDPDA
with $n$ states and with an output size of $\lambda n$, with $\lambda\geq 1$
a fixed rational number, satisfies
$$
f_{\lambda,n}= \gamma_\rho n \stirlingtwo{\rho n}{n}
\dfrac { \left( \lambda+1 \right) ^{(\lambda+1)n\,\rho-1/2}{\beta}^{\lambda\,n\,\rho}}
{{\lambda}^{\lambda\,n\rho+1/2}\sqrt {2\,\pi \,n\,\rho}} (1+o(1)),
$$
\end{proposition}

Note that $f_{\lambda,n}=|\T_{n\rho,n,\lambda\rho n}|$. The following result can be easily obtained.

\begin{proposition}\label{prop:shadow}
For RDPDA in $\T_{s,n,m}$, the average number of pop transitions
is $\frac{s(s-1)}{s+m-1}$.
\end{proposition}

\begin{proof}
We introduce the generating function $G_{s}(z,u)=\left(\dfrac{1}{1-\beta
z}-1+u\right)^s$ counting RDPDA where $z$ counts the output size and $u$ the
number of pop transitions. Using~\cite[Proposition III.2]{FSbook}, the
average number of pop transitions is
$\dfrac{[z^m]\frac{\partial}{\partial
u}G_{s}(z,u)_{|u=1}}{[z^m]G_{s}(z,1)}$. Since $\frac{\partial}{\partial
u}G_{s}(z,u)=s G_{s-1}(z,u)$, the proposition is a direct consequence
of~\eqref{eq:fsn}.
\end{proof}

\subsection{Random Generation}

The random generation is also based on Proposition~\ref{prop:1}; the general
schema for the uniform random generation of an element of $\T_{s,n,m}$
consists of two steps:
\begin{enumerate}
\item Generate uniformly an element $\A$ of $\mathfrak{A}_{s,n}$.
\item Generate the output of the transitions of $\A$ such that the sum of
their sizes is $m$.
\end{enumerate}

The first step can be performed in the general case
using~\cite{DBLP:journals/tcs/HeamNS10}. For generating complete
deterministic RDPDA -- when $s=n\alpha$ -- faster algorithms are described
in~\cite{DBLP:conf/wia/BassinoDN07} and~\cite{DBLP:conf/stacs/CarayolN12}.
In the general case, the complexity is $O(n^3)$ and for the complete case,
the complexity falls to $O(n^{\frac{3}{2}})$. The second step can be easily
done using the classical recursive approach as described in~\cite{DBLP:journals/tcs/FlajoletZC94} or
using Boltzmann samplers.


With a non-optimized Python
implementation running on a 2.5GHz personal computer it is possible to
generate 100 complete RDPDA with hundreds of states in few
minutes. 



\section{Influence of the Accepting Condition}\label{sec:PDA}
 Accessibility defined for a RDPDA does not mean that the accessible states
 can be reached by a calculus. Therefore the random generation may produce
  semantically RDPDA simpler than wanted. One of the
 requirements may be to generate RDPDA accepting non empty languages. 
Another requirement is to produce only reachable states. Finally, if the
final state-empty stack accepting condition is chosen, it is frequently  required
that final states are empty stack reachable.

\subsection{Emptiness of accepted languages}
\begin{proposition}\label{prop:nonvide}
Whatever the selected accepting condition, the probability that an
accessible RDPDA with $n$ states, $s$ transitions and an output size of
$m$, accepts a non empty language is greater or equal to
$\frac{s-1}{2\beta (s+m-1)}$.
\end{proposition}

\begin{proof}
Since the considered RDPDA are  accessible, there is at least one
outgoing transition from the initial state. We will evaluate the probability that
this transition is of the form $(q_{\rm init},a,Z_{\rm init},\varepsilon,p)$
with $p$ final. There is no condition on $a$. The probability that the stack
symbol is $Z_{\rm init}$ is $\frac{1}{\beta}$ since all letters have the
same role. The probability that $p$ is final is $\frac{1}{2}$
(see~\cite{DBLP:journals/tcs/ChamparnaudP05,DBLP:journals/tcs/BassinoN07}).
By Proposition~\ref{prop:shadow} the probability that this transition is a
pop transition is $\frac{s-1}{s+m-1}.$ It follows that the probability
that the transition has the claimed form is $\frac{s-1}{2\beta(s+m-1)}.$ If
this transition exists, the RDPDA accepts the word $a$, proving the
proposition.\hfill$\square$
\end{proof}

By Proposition~\ref{prop:nonvide}, if $m=\lambda s$, for a fixed
$\lambda$, then RDPDA accepting non-empty languages can be randomly
generated by a rejection algorithm, with an expected constant number
of rejects. Experiments show that most of the states are reachable (see
Tables~\ref{tab:xp3} and~\ref{tab:xp4}).

\subsection{Empty-stack Reachability}\label{sec:empty}

Proposition~\ref{prop:nonvide} shows it is possible to generate complete
RDPDA accepting a non-empty language (if $m$ and $s$ are of the same order).
However, it doesn't suffice since many states of generated automata can be
unreachable. Under the final-state and empty-stack condition of acceptance, a
final state that is not reachable with an empty stack is a useless final
state, i.e. it cannot be used as a final state to recognize a word -- but it
can be involved as any state for accepting a word. Using for
instance~\cite{AF-BW-PW-INF-97}, one can decide in polynomial time whether a
state is reachable with an empty stack.

\begin{table}
\begin{center}
\begin{tabular}{|c|c|c|c|c|c|c|c|c|}
\hline
number of states $\to$ & 5& 10 & 15 & 20  & 30 & 40 & 60 & 100\\
\hline 
$\lambda=0.5$ & 3.56 & 6.14 & 8.02 & 11.1 & 16.26&15&24.7& 49.5\\
\hline 
$\lambda= 1$ & 2.6& 4.62 & 4.7 & 6.16 & 7.06 &7.82&13.85&17.3\\
\hline 
$\lambda= 1.5$ & 2.36& 3.05 & 3.61 & 3.62 & 5.2 &5.5&5.68&5.8\\
\hline 
$\lambda= 2$ & 2.0 & 2.6 & 2.81 & 3.4 & 3.02 &3.1&3.21 &3.89\\
\hline 
$\lambda= 3$ & 1.65 & 1.8 & 1.83 & 1.8 & 2.26 &2.44&2.34&2.6\\
\hline 
$\lambda= 5$ & 1.3 & 1.41 & 1.43 & 1.4 & 1.42& 2.1&1.5&1.5\\\hline
\end{tabular}
\end{center}
\caption{Average number of reachable  states with an empty stack,
  $\alpha=2$, $\beta=2$}\label{tab:xp2}
\end{table}

Table~\ref{tab:xp2} reports experiments on the average number of empty-stack
reachable states. For this experiment, we consider complete and accessible
RDPDA with $\alpha=2$, and $\beta=2$. Since a state is final with a
probability $1/2$, dividing the number by 2 in the table provides the
average number of final states reachable with an empty stack. Experiments
show that if $\lambda$ is greater than 1, then the average number of states
reachable with an empty stack is quite small. Remind that $\lambda$ is the
average size of the outputs. For each case, 100 complete
RDPDA have been generated. Clearly the random generation of complete
accessible RDPDA based on the sizes of the output will not produce enough
pop-transitions to empty the stack. Adding a criterion on a minimal number
$k$ of pop-transitions may be a solution that can be achieved in the following
way~:

\begin{enumerate}
\item Choose uniformly $k$ transitions of the underlying finite automata
that will be pop-transitions. 
\item Decorate the $s-k$ other transitions with strings for a total size of $m$.
\end{enumerate}

The experiments reported in Table~\ref{tab:xp2} has also been
done for this procedure to compare the number of rejects. We choose the case
$\alpha=\beta=2$ again with at least 40\% of pop-transitions.
 The results are reported in
Table~\ref{tab:xp6}.

\begin{table}
\begin{center}
\begin{tabular}{|c|c|c|c|c|c|c|c|c|}
\hline
number of states $\to$ & 5& 10 & 15 & 20  & 30 & 40 & 60 & 100\\
\hline 
$\lambda=0.5$ & 4.05 & 7.75 & 11.54 & 16.09 &23.63 & 30.87 & 46.83 &80.62 \\
\hline 
$\lambda= 1$  & 3.09 &5.25  & 8.19 &  10.78&13.97 & 22.55 & 31.23 & 44.52\\
\hline 
$\lambda= 1.5$ & 2.94 & 4.44 & 5.91 & 8.24 & 9.06&  11.93& 18.77 & 29.58 \\
\hline 
$\lambda= 2$ &2.68  & 4.44 & 5.19 & 6.5 & 8.53& 9.39 & 12.84 & 19.91\\
\hline 
$\lambda= 3$ & 2.53 &  3.29& 4.38 &  4.81& 6.37&  7.49&  8.72& 12.18\\
\hline 
$\lambda= 5$  &  2.29& 3.31 & 3.81 & 4.53 &5.09 &  4.71&  5.28& 6.42\\\hline
\end{tabular}
\end{center}
\caption{Average number of reachable states
(at least 40\% of pop transitions); $\alpha=\beta=2$.}\label{tab:xp6}
\end{table}

Imposing a minimal number of pop transitions improves the efficiency
(relative to the number of reachable state) for small values of $\lambda$.
 However it is not sufficient when $\lambda\geq 1$.

\subsection{Reachability (with no stack condition)}

If we are now interested in the random generation with the {\it final state}
condition, regardless of the stack, it is interesting to know the number of
reachable states (which is on average twice the number of final reachable
states). Using~\cite{AF-BW-PW-INF-97}, the average number of reachable
states have been assessed experimentally. Results are reported in
Table~\ref{tab:xp3} for $\alpha=2$ and $\beta=2$ and in Table~\ref{tab:xp4}
for $\alpha=3$ and $\beta=5$.

\begin{table}
\begin{center}
\begin{tabular}{|c|c|c|c|c|c|c|c|c|}
\hline
number of states  $\to$& 10 & 20  & 30 & 40 & 50 & 60 & 80 & 100\\\hline
$\lambda =1$ &  8.29 & 14.89 & 21.5 & 26.93 & 32.48 & 35.55 & 44.4&52.86\\
\hline
$\lambda=2$ & 8.73 & 16.35 & 25.3 & 33.45 & 39.56 & 47.99 & 62.32 & 81.02\\
\hline 
$\lambda=3$ & 8.84 & 17.67 & 27.14 & 36.19 & 45.7 & 54.69 & 73.35 & 89.26\\
\hline 
$\lambda=5$ & 9.23 & 18.3 & 28.06 & 37.47 & 47.61 & 56.7 & 76.11 & 95.15\\
\hline 
\end{tabular}
\end{center}
\caption{Average number of reachable  states, $\alpha=2$ and $\beta=2$.}\label{tab:xp3}
\end{table}

For the random generation of complete accessible RDPDA with a final state
accepting condition,  our framework seems to be  suitable: most of the
states are reachable. For the two other accepting conditions, the value
of $\lambda$ has to be small. Otherwise, there will be too few pop
transitions to clean out the stack. Proposition~\ref{prop:shadow}
confirms this outcome since averagely the number of pop transitions is
close to $\frac{\alpha\beta n}{\lambda+1}$.

\begin{table}
\begin{center}
\begin{tabular}{|c|c|c|c|c|c|c|c|c|}
\hline
number of states  $\to$& 10 & 20  & 30 & 40 & 50 & 60 & 80 & 100\\\hline
$\lambda =1$ &  9.72 & 19.48 & 29.34 & 39.71 & 49.16 & 59.49 &79.6&99.4\\
\hline
$\lambda=2$ & 9.9 & 19.7 & 29.9 & 39.8 & 49.4 & 59.6 & 79.7 & 99.5\\
\hline 
$\lambda=3$ & 9.93 & 19.9 & 29.92 & 39.9 & 49.81 & 59.9 & 79.6 & 99.1\\
\hline 
$\lambda=5$ & 9.95 & 19.96 & 29.93 & 39.9 & 49.86 &59.89  & 79.79 & 99.75\\
\hline 
\end{tabular}
\end{center}
\caption{Average number of reachable  states,  $\alpha=3$ and $\beta=5$.}\label{tab:xp4}
\end{table}

\subsection{A Rejection Algorithm for Reachable Complete RDPDA}

\begin{table}
\begin{center}
\begin{tabular}{|c|c|c|c|c|c|c|c|c|c|}
\hline

&number of states $\to$& 10 & 20  & 30 & 40 & 50 & 60 & 80 & 100\\\hline

&$\lambda=1$& 293.1 & - & -  & - & - & - & - & -\\ 
$\alpha=2$&$\lambda=1.5$& 24.6 & 88.8 & 278.0  & - & - & - & - & - \\
$\beta=2$&$\lambda=2$& 6.9&  20.3& 14.9 & 13.2 &17.9 & 25.2 & 65.9& 95.8\\
&$\lambda=3$&1.6& 1.5& 1.8& 2.0 & 1.9 &2.2 &2.1 &  2.1\\
&$\lambda=5$&1.1& 1.2& 1.1& 1.2 & 1.1 &1.2 &1.2 &  1.2\\
\hline
&$\lambda=1$&  3.8 & 5.5 & 9.4 &15.7  & 38.4& 39.3 &76.9 &76.9\\
$\alpha=4$&$\lambda=1.5$&  1.7&  2.0& 1.7 & 1.8 & 1.9& 2.0 & 2.0&1.9\\
$\beta=2$&$\lambda=2$&  1.3&  1.3& 1.31 &1.3 & 1.2 & 1.2& 1.1 & 1.2\\
&$\lambda=3$&  1.1&  1.1& 1.1 &1.1 & 1.0 & 1.1& 1.1 & 1.1\\
&$\lambda=5$&  1.0&  1.0& 1.0 &1.0 & 1.0 & 1.0& 1.0 & 1.0\\
\hline
& $\lambda=1$& 2.3 & 5.5&16.9 &23.8&20.1&44.1&80.4&214.2\\
$\alpha=2$&$\lambda=1.5$& 1.7& 1.8 &1.2&2.6&3.0&2.0&1.6&1.9\\
$\beta=4$&$\lambda=2$& 1.4& 1.2 & 1.2&1.1&1.3&1.4&1.3&1.1\\
&$\lambda=3$& 1.0 & 1.2& 1.2 &1.0 &1.1&1.1&1.0&1.2\\
&$\lambda=5$& 1.1  &1.0 & 1.0& 1.0&1.0&1.0&1.0&1.0\\
\hline
\end{tabular}
\end{center}
\caption{Average number of random generations to obtain a reachable RDPDA.}\label{tab:xp5}
\end{table}

A natural question is to consider how to generate a complete accessible
RDPDA with exactly $n$ reachable states. An easy way would be to use
a rejection approach by generating a complete accessible RDPDA with $n$
states until obtaining a reachable RDPDA. Results presented in
Tables~\ref{tab:xp3} and~\ref{tab:xp4} seems to prove that this approach
might be fruitful for the parameters of Table~\ref{tab:xp4} but more
difficult for the parameters of Table~\ref{tab:xp3}. Several experiments have been
performed to evaluate the average number of rejects and results are reported
in Table~\ref{tab:xp5}: the average number random generations of RDPDA
used to produce 10 reachable RDPDA was reported. When a ``$-$'' is reported
in the table, it means that after 300 rejects, no such automata was
obtained. These results seem to show that the rejection approach is
tracktable if $\lambda \geq 2$ and if the alphabets are not too small. With
$\alpha=\beta=2$, it works for $\lambda \geq 3$ and for smaller $\lambda$'s
when the number of states is small.

\section{Conclusion}

\begin{table}
\begin{tabular}{|ll|}

\hline Accepting Condition & \\
\hline Empty Stack & 
\begin{tabular}{l} 
$\bullet$ General frameworks does not work: the number of states  \\reachable
  with an empty-stack is too small.\\
$\bullet$ Fixing a minimal number of pop transitions (see
  Section~\ref{sec:empty})\\ works for $\lambda < 1$. 
\end{tabular} \\
\hline 

Final States & 
\begin{tabular}{l}$\bullet$ General framework works: a significant number 
of states are\\ reachable.\\
$\bullet$ A rejection approach is tractable for generating reachable\\ RDPDA,
when both $\lambda \geq 1.5$ and the alphabets are large enough. 
\end{tabular}\\
\hline

\end{tabular}\\

\caption{Random Generation of RDPDA.}\label{tab:conclusion}
\end{table}

In this paper a general framework for generating accessible deterministic
pushdown automata is proposed. We also experimentally showed that with some
accepting conditions, it is possible to generate pushdown automata where
most states are reachable. The results on the random generation are
synthesized in Table~\ref{tab:conclusion}. In a future work we plan to investigate how
to randomly generate real-time deterministic automata, with an empty-stack
accepting condition and again, which most states are reachable. We also plan
to remove the {\it real-time} assumption, but it requires a deeper work on
the underlying automata.

\bibliographystyle{alpha}
\bibliography{pda}

\newpage


\end{document}